\def\draft{1}
\documentclass[11pt]{article}

\usepackage[margin=1in]{geometry}
\usepackage[utf8]{inputenc}
\usepackage{bm}
\usepackage{graphicx}
\usepackage{color,xcolor}
\usepackage{amsmath,amsfonts, amssymb, amsthm, thmtools}
\usepackage{thm-restate}
\usepackage{algorithm, algpseudocode}
\usepackage{caption}
\usepackage{subcaption}

\PassOptionsToPackage{hyphens}{url}
\usepackage[colorlinks=true, allcolors=blue]{hyperref}
\usepackage[capitalise,nameinlink]{cleveref}

\usepackage[style=alphabetic, backend=biber, minalphanames=3, maxalphanames=4, maxbibnames=99]{biblatex}

\crefformat{section}{#2\S#1#3}
\crefformat{subsection}{#2\S#1#3}
\crefformat{subsubsection}{#2\S#1#3}

\declaretheoremstyle[bodyfont=\it,qed=\qedsymbol]{noproofstyle}

\numberwithin{equation}{section}

\declaretheorem[name=Observation,numbered=no]{observation*}

\declaretheorem[name=Theorem,numberwithin=section]{theorem}

\declaretheorem[name=Theorem,numbered=no]{theorem*}

\declaretheorem[name=Lemma,numbered=no]{lemma*}

\declaretheorem[name=Corollary,numbered=no]{corollary*}

\declaretheorem[name=Proposition,numbered=no]{proposition*}

\declaretheorem[name=Claim,numbered=no]{claim*}

\declaretheorem[name=Conjecture,numbered=no]{conjecture*}

\declaretheorem[name=Question,numbered=no]{question*}

\declaretheoremstyle[bodyfont=\it]{defstyle} 

\declaretheorem[numberlike=equation,style=defstyle]{definition}
\declaretheorem[unnumbered,name=Definition,style=defstyle]{definition*}

\declaretheorem[unnumbered,name=Example,style=defstyle]{example*}

\declaretheorem[unnumbered,name=Notation=defstyle]{notation*}

\declaretheorem[unnumbered,name=Construction,style=defstyle]{construction*}

\declaretheoremstyle[]{rmkstyle} 

\newtheorem*{remark}{Remark}

\newcommand{\RoundPLSig}{\mathsf{PLSigmoid}}
\newcommand{\RoundS}{\mathsf{S}}

\newcommand{\mcut}{\textsc{Max-Cut}}
\newcommand{\mdcut}{\textsc{Max-DiCut}}
\newcommand{\mkand}{\textsc{Max-}k\textsc{And}}

\newcommand{\Obl}[1]{\mathcal{O}_{#1}}

\renewcommand{\deg}{\mathsf{deg}}
\newcommand{\dout}{\mathsf{outdeg}}
\newcommand{\din}{\mathsf{indeg}}
\newcommand{\val}{\mathsf{val}}
\newcommand{\bias}{\mathsf{bias}}

\usepackage{singer-macros}

\usepackage{tikz}
\usetikzlibrary{arrows.meta}
\usetikzlibrary{decorations.markings}

\definecolor{figblue}{HTML}{ABDEE6}
\definecolor{figred}{HTML}{FEE1E8}
\definecolor{figpurp}{HTML}{CBAACB}
\definecolor{figdarkblue}{HTML}{8FAAE3}

\def\FIGcolpos{figblue}
\def\FIGcolposdesc{~\colorbox{\FIGcolpos}{LIGHT BLUE}~}
\def\FIGcolneg{figred}
\def\FIGcolnegdesc{~\colorbox{\FIGcolneg}{PINK}~}
\def\FIGcolzero{figpurp}
\def\FIGcolzerodesc{~\colorbox{\FIGcolzero}{PURPLE}~}
\def\FIGcolsuperpos{figdarkblue}
\def\FIGcolsuperposdesc{~\colorbox{\FIGcolsuperpos}{DARK BLUE}~}
\def\FIGcollab{gray!20!white}

\def\FIGpackx{3}
\def\FIGpacky{3}
\def\FIGbend{25}
\def\FIGsubfigwidth{0.45\textwidth}

\tikzset{weight/.style = {font=\small, fill=\FIGcollab, draw=black}}
\tikzset{
    diredge/.style = {line width=0.3mm, arrows={-Latex[angle=60:2.5mm]}}
}

\usepackage{pgfplots}

\usepackage{empheq}
\usepackage{tcolorbox}

\title{Oblivious Algorithms for Maximum Directed Cut: New Upper and Lower Bounds}
\author{\ifnum\draft=1
Samuel Hwang\thanks{Harvard College, Harvard University, Cambridge, 
MA, USA. Email: \texttt{samuelhwang@college.harvard.edu}.}
\and Noah G. Singer\thanks{Department of Computer Science, Carnegie Mellon University, Pittsburgh, PA, USA. Supported by an NSF Graduate Research Fellowship (Award DGE2140739). Email: \texttt{ngsinger@cs.cmu.edu}.}
\and Santhoshini Velusamy\thanks{Toyota Technological Institute at Chicago, Chicago, IL, USA. Supported by an NSF CRII award CCF 2348475. Email: \texttt{santhoshini@ttic.edu}.}
\fi}

\begin{document}
\date{}
\sloppy
\maketitle


\begin{abstract}
In the \emph{maximum directed cut} ($\mdcut$) problem, the input is a directed graph $G=(V,E)$, and the goal is to pick a partition $V = S \cup (V \setminus S)$ of the vertices such that as many edges as possible go \emph{from} $S$ \emph{to} $V\setminus S$. \emph{Oblivious algorithms}, introduced by \textcite{FJ15}, are a simple class of algorithms for this problem. These algorithms independently and randomly assign each vertex $v$ to either $S$ or $V \setminus S$, and the distribution of $v$'s assignment is determined using only extremely local information about $v$: its \emph{bias}, i.e., the relative difference between its out- and in-degrees. These algorithms have natural implementations in certain graph streaming models, where they have important implications \cite{SSSV23-dicut,SSSV23-random-ordering,kallaugher2023exponential}.

In this work, we narrow the gap between upper and lower bounds on the best approximation ratio achievable by oblivious algorithms for $\mdcut$. We show that there exists an oblivious algorithm achieving an approximation ratio of at least $0.4853$, while every oblivious algorithm obeying a natural symmetry property achieves an approximation ratio of at most $0.4889$. The previous known bounds were $0.4844$ and $0.4899$, due to \textcite{Sin23-kand,FJ15}, respectively. Our techniques involve designing principled parameterizations of the spaces of algorithms and lower bounds and then executing computer searches through these spaces.
\end{abstract}

\section{Introduction}

In this work, we study a special class of algorithms, called \emph{oblivious algorithms}, for a specific constraint satisfaction problem, called \emph{maximum directed cut} ($\mdcut$). We first informally describe these two notions; see \cref{sec:prelim:mdcut} below for formal definitions.

\subsection{Background: Directed cuts, bias, and oblivious algorithms}

An input instance to $\mdcut$ is a \emph{directed graph} $G = ([n],E)$ (possibly with edge weights) on vertex set $[n] = \{1,\ldots,n\}$. A \emph{cut} is a vector $\vecx = (x_1,\ldots,x_n) \in \{\pm 1\}^n$, and can be viewed as an assignment of a bit to each vertex in the graph. An edge $(v_1,v_2)$ is \emph{satisfied} by the cut $\vecx$ if $x_{v_1} = 1$ and $x_{v_2} = 0$. The \emph{value} of a cut is the fraction of edges it satisfies, and the \emph{value} of a graph, denoted $\val_G$, is the maximum value over all cuts.

To define oblivious algorithms, we first define a scalar quantity associated to each vertex in a directed graph called \emph{bias}. The bias of a vertex is simply \[ \bias_G(v) \eqdef \frac{\dout_G(v)-\din_G(v)}{\dout_G(v)+\din_G(v)}, \] where $\dout_G(v)$ and $\din_G(v)$ are respectively the out- and in-degrees of $v$ in $G$. Note that $\bias_G(v)$ ranges from $+1$ ($v$ has only out-edges) to $-1$ ($v$ has only in-edges).

Oblivious algorithms are a class of randomized algorithms for $\mdcut$ which ``only know'' about the bias of each vertex. More formally, an oblivious algorithm is defined by a so-called \emph{selection function} $\RoundS : [-1,+1] \to [0,1]$. The corresponding algorithm $\Obl{\RoundS}$, given a graph $G$, produces an assignment by independently setting each vertex $v$ to equal $1$ w.p. $\RoundS(\bias_G(v))$ and $0$ otherwise. The \emph{approximation ratio} of $\Obl{\RoundS}$, denoted $\alpha(\Obl{\RoundS})$, is the ratio of the expected value of this ``oblivious assignment'' to the value of the best assignment, minimized over all graphs $G$.

Oblivious algorithms for $\mdcut$ were introduced by \textcite{FJ15}, who proved both upper and lower bounds\footnote{Note on language: In this paper, we use the usual convention that upper bounds are algorithms and lower bounds are hardness results. Confusingly, an ``upper bound'' actually lower-bounds the maximum ratio achievable by any oblivious algorithm, while a ``lower bound'' upper-bounds this ratio.} on their capacity to approximate $\mdcut$:

\begin{theorem}[Prior upper bound, {\cite[Thm. 1.3]{FJ15}}]\label{thm:fj-ub}
    There exists an oblivious algorithm $\Obl{\RoundS}$ achieving an approximation ratio $\alpha(\Obl{\RoundS}) \geq 0.483$.
\end{theorem}

A selection function $\RoundS : [-1,+1] \to [0,1]$ is \emph{antisymmetric} if for all $b \in [-1,+1]$, $\RoundS(-b) = 1-\RoundS(b)$ (see also \cref{def:antisym}). Most of the selection functions which have been studied have this property.

\begin{restatable}[Prior lower bound for antisymmetric selection, {\cite[Thm. 1.4]{FJ15}}]{theorem}{thmfjantisym}\label{thm:fj-lb-antisym}
    If $\RoundS : [-1,+1] \to [0,1]$ is an antisymmetric selection function, then the oblivious algorithm $\Obl{\RoundS}$ achieves an approximation ratio $\alpha(\Obl{\RoundS}) \leq 0.4899$. 
\end{restatable}

We include a brief proof of \cref{thm:fj-lb-antisym} in \cref{sec:fj-lb} to facilitate comparison with our new lower bounds.

\begin{theorem}[Prior lower bound for general selection, {\cite[Thm. 1.5]{FJ15}}]\label{thm:fj-lb-general}
    If $\RoundS : [-1,+1] \to [0,1]$ is any selection function, then the oblivious algorithm $\Obl{\RoundS}$ achieves an approximation ratio $\alpha(\Obl{\RoundS}) \leq 0.4998$. 
\end{theorem}

\begin{remark}
    The constant in \cref{thm:fj-lb-general} is not optimized in the work of \textcite{FJ15}. However, it is straightforward to compute the best possible constant achievable with their proof technique. Indeed, the proof of \cite[Thm. 1.5]{FJ15} considers two graphs $G$ and $L$, and shows that if $\RoundS$ is any selection function with $\RoundS(\frac12) = \frac12+\delta$, then $\RoundS$ achieves ratio at most $0.4899 + \delta$ on $G$ and $\frac12-2\delta^2$ on $L$; hence, $\RoundS$ achieves ratio strictly below $\frac12$ on either $G$ or $L$. But equating these two quantities and solving for $\delta$ yields $\delta = 0.0099$; at this point, both quantities are roughly $0.4998$. At any other value of $\delta$, at least one of the quantities will be larger, and therefore, $0.4998$ is the optimal bound.
\end{remark}

One particularly nice feature of oblivious algorithms described by \cite{FJ15} is that (if the selection function $\RoundS$ is piecewise constant) the approximation ratio of $\Obl{\RoundS}$ can be calculated by a simple \emph{linear program} (LP); indeed, this LP essentially ``encodes'' the process of minimizing the approximation ratio over all possible graphs to find the worst-case input (see \cref{thm:antisymmetric-lp} below). A more recent work of \textcite{Sin23-kand} provided an open-source \texttt{python} implementation of the ratio-calculating LP, and used it to improve \cref{thm:fj-ub} via a more refined analysis of a similar oblivious algorithm:

\begin{theorem}[Improved prior upper bound, {\cite{Sin23-kand}}]\label{thm:sin-ub}
    There exists an oblivious algorithm $\Obl{\RoundS}$ achieving an approximation ratio $\alpha(\Obl{\RoundS}) \geq 0.484$.
\end{theorem}

However, both the works \cite{FJ15,Sin23-kand} left a significant open question:

\begin{question*}
    What is the best possible approximation ratio $\alpha(\Obl{\RoundS})$ which can be achieved by any oblivious algorithm $\Obl{\RoundS}$ for $\mdcut$?
\end{question*}

As we describe in \cref{sec:motivations} below, oblivious algorithms for $\mdcut$ are known to imply algorithms achieving (arbitrarily close to) the same ratio in several different streaming models \cite{SSSV23-dicut,SSSV23-random-ordering,kallaugher2023exponential}. Resolving this question would therefore characterize the best approximation ratios achievable via these streaming techniques. In this work, we make progress on the question and tighter upper and lower bounds on oblivious algorithms via intricate computer searches.

\subsection{Results}

We define a class of selection functions, which we call \emph{piecewise linear \emph{(PL)} sigmoid} functions (see \cref{def:pl-sig} below). These functions are denoted $\RoundPLSig_b$, where $b \in [0,1]$ is an ``intercept'' parameter. \cref{thm:fj-ub,thm:sin-ub} were proven by analyzing discretizations (i.e., piecewise-constant versions) of $\RoundPLSig_{1/2}$. We demonstrate a strictly better oblivious algorithm, using a discretization of a PL sigmoid function with a different intercept, namely, $\RoundPLSig_{149/309}$:

\begin{restatable}[New upper bound]{theorem}{thmplsigub}\label{thm:pl-ub}
    There exists an oblivious algorithm $\Obl{\RoundS}$ achieving an approximation ratio $\alpha(\Obl{\RoundS}) \geq 0.485359$.
\end{restatable}

We also complement this theorem with lower bounds, both for $\RoundPLSig_{1/2}$ itself and for arbitrary PL sigmoid functions:

\begin{restatable}[Lower bound for PL sigmoid selection with $b=1/2$ intercept]{theorem}{thmhalf}\label{thm:intercept-1/2-lb}
    $\alpha(\Obl{\RoundPLSig_{1/2}}) \leq 0.485282$.
\end{restatable}

\begin{restatable}[Lower bound for PL sigmoid selection with arbitrary intercept]{theorem}{thmplsiglb}\label{thm:pl-lb}
    For every $b \in [0,1]$, $\alpha(\Obl{\RoundPLSig_b}) \leq 0.486$.
\end{restatable}

Note that these two lower bounds hold for PL sigmoid functions themselves (which are continuous), while the upper bound \cref{thm:pl-ub} used a discretization of such a function, so they are not formally comparable. In particular, we observe that given the approximation ratio seems to increase and converge to a limit during finer discretization (see \cref{fig:discretization} below). In light of this, we note that the ratio in this theorem, $0.485282$, is only slightly larger than the ratio from \cref{thm:pl-ub}, $0.485275$. This is heuristic evidence that the selection function $\RoundPLSig_{149/309}$ has a strictly higher approximation ratio than the selection function $\RoundPLSig_{1/2}$.

We prove another, more general, lower bound that holds against antisymmetric selection functions:

\begin{restatable}[Lower bound for symmetric selection]{theorem}{thmsymlb}\label{thm:antisym-lb}
    For every \emph{antisymmetric} selection function $\RoundS$, $\Obl{\RoundS}$ achieves an approximation ratio $\alpha(\Obl{\RoundS}) \leq 0.4889$.
\end{restatable}

\cref{thm:antisym-lb} should be compared against \cref{thm:fj-lb-antisym}, which it improves by roughly $0.001$. Finally, we prove a lower bound against arbitrary (not necessarily antisymmetric) selection functions:

\begin{restatable}[Lower bound for general selection]{theorem}{thmgenlb}\label{thm:gen-lb}
    For \emph{every} selection function $\RoundS$, $\Obl{\RoundS}$ achieves an approximation ratio $\alpha(\Obl{\RoundS}) \leq 0.4955$.
\end{restatable}

The proof of this theorem is itself a conceptual contribution: The lower bound is witnessed by a single, simple graph, with only one bias (up to sign). In contrast, \cref{thm:fj-lb-antisym} was weaker (by $0.003$), and its proof used multiple graphs, one of which had two distinct biases appearing (up to sign).

\subsection{Motivations}\label{sec:motivations}

\paragraph{Downstream applications.} \textcite{FJ15} were interested in oblivious algorithms both in their own right as a nontrivial class of combinatorial algorithms for $\mdcut$ and because of connections with ``local'' and ``distributed'' models of computation. But more recently, several works have established that the existence of good oblivious algorithms for $\mdcut$ implies the existence of certain kinds of good streaming algorithms for $\mdcut$. These algorithms are given a list of the graph's directed edges as input and must output an estimate of its $\mdcut$ value. In particular, the existence of $\alpha$-approximation oblivious algorithms for $\mdcut$ is known to imply $(\alpha-\epsilon)$-approximation algorithms for all $\epsilon > 0$ in the following models:

\begin{table}[h]
    \centering
    \begin{tabular}{ccccc}
        \textbf{Space} & \textbf{Input ordering} & \textbf{\# passes} & \textbf{Setting} & \textbf{Citation} \\ \hline
        $O(\log n)$                 & Random        & $1$ & Classical & \cite{SSSV23-random-ordering} \\
        $O(\log n)$                 & Adversarial   & $2$ & Classical & \cite{SSSV23-random-ordering} \\
        $O(\sqrt n \polylog n)$     & Adversarial   & $1$ & Classical & \cite{SSSV23-dicut} \\
        $O(\polylog n)$             & Adversarial   & $1$ & Quantum & \cite{kallaugher2023exponential} \\
    \end{tabular}
    \caption{Streaming models into which oblivious algorithms are known to ``translate'', achieving the same approximation ratio up to arbitrarily small constants.}
\end{table}

Thus, further improved oblivious algorithms, like we provide in this paper, imply further improvements in the state-of-the-art for all of these streaming models.

The fact that oblivious algorithms can be ``implemented'' as streaming algorithms in these models is motivated by lower bounds known in some related models. In particular, it was known by a previous result of \textcite{CGV20} that $(4/9-\epsilon)$-approximations to $\mdcut$ can be computed using $O(\log n)$ space in a single, classical, adversarially-ordered pass, while $(4/9+\epsilon)$-approximations require $\Omega(\sqrt n)$ space for a single, classical, adversarially-ordered pass. Thus, the fact that \textcite{FJ15} constructed oblivious algorithms achieving an approximation ratio $0.483 > 4/9$ implied that the \cite{CGV20} lower bound was ``tight'': Adjusting the model to add either polylogarithmically more space, random ordering, a second pass, or access to quantum bits yields strictly better approximations! In contrast, for $\mdcut$'s ``undirected cousin'' $\mcut$, optimal lower bounds are known even for $O(\sqrt n)$-space random-ordering algorithms and $o(n)$-space adversarial-ordering algorithms \cite{KKS15,KK19}. In the quantum setting, \textcite{kallaugher2023exponential} claim that $\mdcut$ is the first discrete optimization problem with a provable exponential separation in complexity between classical and quantum algorithms. These separations are all powered by the existence of oblivious algorithms achieving a ratio strictly above $4/9$. 

\paragraph{Oblivious algorithms for other problems.} Oblivious algorithms for $\mdcut$ have also been studied in several other areas. In mechanism design, \textcite{Luk14} showed that any \emph{monotone}\footnote{An oblivious algorithm is said to be monotone if its corresponding selection function $\RoundS$ is monotone, i.e., $\RoundS(x)\le \RoundS(y)$ if and only if $x\le y$.} oblivious algorithm is a \emph{strategy-proof mechanism}\footnote{A mechanism for $\mdcut$ is defined as follows. There are $m$ players, each given a unique edge of the directed graph. A mechanism is a function that asks each player to reveal their edge and then chooses a (random) assignment of the vertices. The player may or may not reveal their edge truthfully. The utility of each player is the expected probability that their edge is satisfied by this assignment. A mechanism is said to be strategy-proof if the optimal strategy for every player is to be truthful.} for $\mdcut$. \textcite{BFS19} showed applications of oblivious algorithms to the online submodular optimization problem.

\textcite{Sin23-kand} recently extended the definition of oblivious algorithms to a more general version of $\mdcut$ called $\mkand$, where each constraint applies to $k$ variables and specifies a single required bit for each variable. ($\mdcut$ is the special case where $k=2$ and in each constraint, exactly one variable needs to be assigned $1$ and the other $0$.) He also generalized the LP of \cite{FJ15} for calculating the ratio of a piecewise-constant oblivious $\mdcut$ algorithm to $\mkand$, and used this to achieve streaming separation results \`a la \cite{SSSV23-random-ordering}.

\subsection{Structure of rest of the paper}

\cref{sec:preliminaries} contains some of the preliminary background used in the rest of the paper. In \cref{sec:pl-ub}, we discuss our new selection function that achieves an approximation ratio of at least 0.485275. In \cref{sec:plsigmoid_half}, we prove \cref{thm:intercept-1/2-lb} by explicitly constructing a graph (only on two vertices!) for which the $\RoundPLSig_{1/2}$ function achieves a strictly weaker than $0.485282$ approximation. In \cref{sec:lowerbound_plsigmoid}, we prove \cref{thm:pl-lb} by constructing three graphs (one of which is on forty vertices!) and showing that every $\RoundPLSig$ function has an approximation ratio of at most $0.486$ on at least one of these graphs. We also give a detailed description of the linear program that we use to generate these graphs. In \cref{sec:lowerbound_general}, we prove \cref{thm:gen-lb} using a pair of graphs (a two-vertex and a four-vertex graph). Finally, in \cref{sec:lowerbound_antisymmetric}, we prove \cref{thm:antisym-lb} and give a detailed description of the methodology we use to construct the lower bound instance.

\subsection{Code}

All code for this paper is available on GitHub at \url{https://github.com/singerng/oblivious-csps/}.

\section{Preliminaries}\label{sec:preliminaries}

We begin with some basic notations for directed graphs and for oblivious algorithms.

\subsection{Directed graphs}

\begin{definition}[Directed graphs]
    A \emph{(weighted) directed graph} $G$ is a pair $(V,w)$ where $V$ is a set of \emph{vertices} and $w : V \times V \to \BR_{\geq 0}$ is a \emph{weight function} satisfying $w(v,v) = 0$ for all $v \in V$.
\end{definition}

Oftentimes, we are interested in weighted graphs with integer weights. In this setting, it is useful to think of (multi)graphs as pairs $(V,E)$, where $E \subseteq V \times V$ is a multiset of \emph{edges} (again, $(v,v) \not\in V$ for all $v \in V$); the set $E$ induces the weighted graph where $w(v_1,v_2)$ is the multiplicity of $(v_1,v_2)$ in $E$.

\begin{definition}[Degree]
    Let $G = (V, w)$ be a directed graph. For a vertex $v \in V$, the \emph{out-, in-}, and \emph{total degrees} of $v$ are \[ \dout_G(v) \eqdef \sum_{u \in V} w(v,u), \din_G(v) \eqdef \sum_{u \in V} w(u,v), \text{ and } \deg_G(v) \eqdef \dout_G(v) + \din_G(v). \] A vertex $v$ is \emph{isolated} if $\deg_G(v) = 0$.
\end{definition}

We also let $m_G \eqdef \sum_{v_1 \neq v_2 \in V} w(v_1,v_2)$ denote the total weight in $G$.

\begin{definition}[Bias]
    Let $G = (V,w)$ be a graph, and $v \in V$ a nonisolated vertex. Then the \emph{bias} of $v$ is \[ \bias_G(v) \eqdef \frac{\dout_G(v) - \din_G(v)}{\deg_G(v)}. \]
\end{definition}

Observe that $-1 \leq \bias_G(v) \leq +1$.

\subsection{$\mdcut$ and oblivious algorithms}\label{sec:prelim:mdcut}

\begin{definition}[$\mdcut$]
The \emph{maximum directed cut} ($\mdcut$) problem is defined as follows: The input is a directed graph $G = (V,w)$. For any \emph{assignment} (a.k.a. \emph{cut}) $\vecx = (x_v)_{v \in V} \in \{0,1\}^V$, the \emph{$\mdcut$ value} of $\vecx$ is \[ \val_G(\vecx) \eqdef \frac1{m_G} \sum_{v_1 \neq v_2 \in V} w(v_1,v_2) \cdot \1[x_{v_1} = 1 \text{ and } x_{v_2} = 0]. \] The \emph{$\mdcut$ value} of $G$ is \[ \val_G \eqdef \max_{\vecx \in \{0,1\}^V} \val_G(\vecx). \] The goal of the $\mdcut$ problem is to approximate $\val_G$ given $G$.
\end{definition}

Now, we consider algorithms for $\mdcut$ which estimate $\val_G$ using only the biases of vertices in $G$:

\begin{definition}
    An \emph{oblivious algorithm} for $\mdcut$ is defined by a \emph{selection function} $\RoundS : [-1,+1] \to [0,1]$. The corresponding algorithm, denoted $\Obl{\RoundS}$, behaves as follows. Given a directed graph $G = (V,w)$, $\Obl{\RoundS}$ outputs \[ \Obl{\RoundS}(G) := \frac1m \sum_{v_1 \neq v_2 \in V} w(v_1, v_2) \cdot (\RoundS(\bias_G(v_1))) (1-\RoundS(\bias_G(v_2))). \]
\end{definition}

Observe that $\CA^\RoundS(G)$ equals the expected value of the random cut which assigns each (nonisolated) vertex $v$ as an independent $\Bern(\bias_G(v))$ random variable. Thus, $\CA^\RoundS(G) \leq \val_G$. We are interested in how good of an approximation the algorithm provides:

\begin{definition}[Approximation]
    Let $G$ be a graph and $\Obl{\RoundS}$ an oblivious algorithm. The \emph{approximation ratio} of $\Obl{\RoundS}$ on $G$ is \[ \alpha(\Obl{\RoundS};G) \eqdef \frac{\Obl{\RoundS}(G)}{\val_G}. \] The \emph{approximation ratio} of $\Obl{\RoundS}$ is \[ \alpha(\Obl{\RoundS}) \eqdef \inf_{\text{graph } G} \alpha(\Obl{\RoundS};G). \]
\end{definition}

In this paper, we are interested in some specific types of selection functions.

\begin{definition}[Antisymmetry]\label{def:antisym}
    A selection function $\RoundS : [-1, +1] \to [0, 1]$ is \emph{antisymmetric} if $\RoundS(x) = 1 - \RoundS(-x)$ for all $x \in [-1, +1].$
\end{definition}

Notably, if $\RoundS$ is antisymmetric, then $\RoundS(0) = 1/2$. Antisymmetry is a natural desideratum for selection functions for $\mdcut$, since the operation of flipping all edges (i.e., switching to the ``transpose'' weight function $w^\top(v_1,v_2) = w(v_1,v_2)$) preserves the $\mdcut$ value; this operation also negates the bias of every vertex, and so antisymmetry implies that the output of the oblivious algorithm is also preserved.

Another class of selection functions we are interested in is ``piecewise constant'' functions:

\begin{definition}
    A selection function $\RoundS : [-1,+1] \to [0,1]$ is \emph{$\ell$-class piecewise constant} if there exists a partition of the domain $[-1,+1]$ into $\ell$ intervals $I_1,\ldots,I_\ell$, and $\ell$ values $p_1,\ldots,p_\ell$, such that $\RoundS(x) = p_i$ for all $x \in I_i$.
\end{definition}

\subsection{Linear program}

For completeness, we include here the linear program introduced by \cite{Sin23-kand} for calculating the approximation ratio of an antisymmetric piecewise constant selection function. For $\ell \in \mathbb{N}$, let $[\pm\ell] = \{-\ell,\ldots,+\ell\}$.

\begin{theorem}[LP for antisymmetric selection functions]\label{thm:antisymmetric-lp}
    Let $\ell \in \mathbb{N}$. Let $0 \leq t_0 \leq \cdots \leq t_\ell = 1$, and $0 \leq p_1,\ldots,p_\ell \leq 1$. Define intervals $I_{+i} := (+t_{i-1},+t_i]$ and $I_{-i} := [-t_i,-t_{i-1})$ for $i \in [\ell]$ and $I_0 := [-t_0,+t_0]$. Let $\RoundS : [-1,+1] \to [0,1]$ be the antisymmetric selection function which maps $I_{+i}$ to $p_i$ and $I_{-i}$ to $1-p_i$ for $i \in [\ell]$ and $I_0$ to $\frac12$. Then the approximation ratio $\alpha(\Obl{\RoundS})$ achieved by $\Obl{\RoundS}$ equals the value of the following linear program:

\begin{tcolorbox}
\begin{empheq}[left=\empheqlbrace]{alignat*=3}
    & \underset{\{w(\vecc) : \vecc \in C\}}{\mathrm{minimize}} \quad && \sum_{\vecc \in C} p(\vecc) w(\vecc) && \\
    & \mathrm{s.t.} && w(\vecc) \geq 0 && \forall \vecc \in C \\
    & && \sum_{\vecc \in C^+} w(\vecc) = 1 && \\
    & && b_i (W^+(i) + W^-(i)) \leq W^+(i) - W^-(i) \quad && \forall i \in [\pm\ell] \\
    & && W^+(i) - W^-(i) \leq a_i (W^+(i) + W^-(i)) \quad && \forall i \in [\pm\ell]
\end{empheq}
\end{tcolorbox}
where we define the set of $\BN$-valued matrices:
\[ C = \left\{ \vecc \in \BN^{\{\pm1\} \times [\pm \ell]} : \sum_{b \in \{\pm 1\}} \sum_{i \in [\pm \ell]} c_{b,i} = 2 \right\}; \]
the subset of $C$ supported on the $+1$ row:
\[ C^+ = \{ \vecc \in C : c_{-1,i} = 0 \quad \forall i \in [\pm \ell] \}; \]
a constant for each vector in $C$:
\[ p(\vecc) = \left(\frac12\right)^{c(+1,0) + c(-1,0)} \prod_{i \in [\ell]} p_i^{c(+1,+i) + c(-1,-i)} (1-p_i)^{c(+1,-i) + c(-i,+i)}; \]
constants for $i \in [\pm \ell]$:
\[ a_i = \sup I_i \text{ and } b_i \inf I_i; \]
and the linear functions on $\{w(\vecc)\}$, for $i \in [\pm \ell]$: \[ W^+(i) = \sum_{\vecc \in C} c_i^+ w(\vecc) \text{ and } W^-(i) = \sum_{\vecc \in C} c_i^- w(\vecc) . \]
\end{theorem}

We note that \textcite{FJ15} gave an alternative (but highly similar) LP which worked for arbitrary selection functions --- i.e., not just antisymmetric ones --- but required roughly twice as many variables. In this paper, we only run the LP for antisymmetric selection functions, so we use the LP from \cite{Sin23-kand} because of the cost savings.

\section{Improved oblivious algorithms (\cref{thm:pl-ub})}\label{sec:pl-ub}

Our goal in this section is to prove \cref{thm:pl-ub}:

\thmplsigub*

To prove this theorem, we introduce a specific type of antisymmetric selection function we are interested in, namely, an ``\emph{S}-shaped'' piecewise linear function:

\begin{definition}[PL sigmoid functions]\label{def:pl-sig}
    A \emph{piecewise linear \emph{(PL)} sigmoid function} is a selection function of the following form: For an \emph{intercept} parameter $b \in [0,1]$, \[ \RoundPLSig_b(x) = \begin{cases} 0 & x \leq -b \\ 1/2 + \frac{x}{2b} & -b \leq x \leq +b \\ 1 & x \geq +b. \end{cases} \]
\end{definition}

Note that according to $\RoundPLSig_b$, vertices with bias exceeding $b$ in magnitude are assigned deterministically (i.e., they are always assigned to $+1$ or $-1$, depending on their sign), and vertices with smaller bias interpolate linearly between these two extremes. See \cref{fig:pl-sigmoid} for a visual depiction of a PL sigmoid and its discretization.

\def\FIGlo{4}
\def\FIGb{.5}

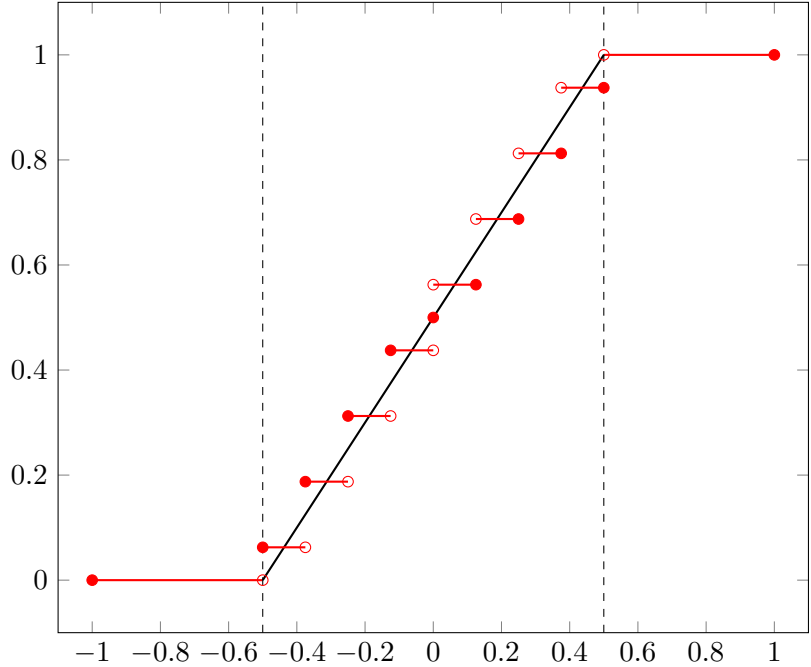
\begin{figure}[t!]
\centering
\begin{tikzpicture}
\begin{axis}[
    width=0.7\textwidth,
    xmin=-1.1, xmax=+1.1,
    ymin=-0.1, ymax=1.1,
    legend pos=south east
]

\addplot[mark=none, red, thick] coordinates {(-1,0) (-\FIGb,0)};

\addplot[only marks, mark=*, color=red] coordinates {(0,1/2)};
\addplot[only marks, mark=o, color=red] coordinates {(\FIGb,1)};
\addplot[only marks, mark=o, color=red] coordinates {(-\FIGb,0)};
\addplot[only marks, mark=*, color=red] coordinates {(1,1)};
\addplot[only marks, mark=*, color=red] coordinates {(-1,0)};

\addplot[mark=none, red, thick] coordinates {(\FIGb,1) (1,1)};

\addplot[mark=none, black, thick] coordinates {(-\FIGb,0) (\FIGb,1)};

\addplot[mark=none, black, dashed] coordinates {(-\FIGb,-.1) (-\FIGb,1.1)};
\addplot[mark=none, black, dashed] coordinates {(\FIGb,-.1) (\FIGb,1.1)};

\foreach \i in {1,...,\FIGlo} {
    \pgfmathsetmacro{\smallx}{((\i-1) / \FIGlo * \FIGb}
    \pgfmathsetmacro{\bigx}{\i / \FIGlo * \FIGb}
    \pgfmathsetmacro{\samey}{1/2 + 1/8 * (2 * \i - 1) / \FIGlo / \FIGb}

    \addplot[only marks, mark=o, color=red] coordinates {(\smallx,\samey)};
    \addplot[only marks, mark=*, color=red] coordinates {(\bigx,\samey)};
    \addplot[mark=none, red, thick] coordinates {(\smallx,\samey) (\bigx, \samey)};

    \addplot[only marks, mark=o, color=red] coordinates {(-\smallx,1-\samey)};
    \addplot[only marks, mark=*, color=red] coordinates {(-\bigx,1-\samey)};
    \addplot[mark=none, red, thick] coordinates {(-\smallx,1-\samey) (-\bigx, 1-\samey)};
}

\end{axis}
\end{tikzpicture}
\caption{The step function $\RoundPLSig_{1/2}$ and its discretization into $\ell=5$ positive bias classes. The discretization is the function in red. The jump discontinuities are notated using standard marks: Open circles are open interval ends and closed circles are closed interval ends. The continuous (non-discretized) function $\RoundPLSig_{1/2}$ disagrees with its discretization only within the interval $[-1/2,+1/2]$ (marked by the vertical dashed line segments). The continuous function is represented by the black line segment within this interval.}\label{fig:pl-sigmoid}
\end{figure}

In the prior works \cite{FJ15,Sin23-kand}, the highest approximation ratios achieved by oblivious algorithms were found by using discretized versions of $\RoundPLSig_{1/2}$. These discretizations --- in the antisymmetric case --- are controlled by a parameter $\ell$, denoting the number of \emph{positive} bias classes. (There are $L = 2\ell+1$ bias classes in general.) \textcite{Sin23-kand} used $\ell=200$, though it appears that his discretization was uniform and so it split the interval $[1/2,1]$ into $\approx 100$ bias classes, and we condense them into only one bias class. We also use a finer discretization with $\ell = 251$. Most importantly, we use a different intercept parameter $b$. To choose $b$, we performed a binary search using a discretization with $\ell = 51$. We found that an intercept at $b = 149/309$ gives the best bound among the intercepts we inspected.

\begin{proof}[Proof of \cref{thm:pl-ub}]
     We set $b=149/309$ and use a discretization of $\RoundPLSig_b$ with $\ell = 251$ classes. We plug this discretization into the linear program for calculating approximation ratios for antisymmetric functions (\cref{thm:antisymmetric-lp}). Ours is the natural discretization: We set $t_i = \frac{i}{\ell-1} b$ for $i \in \{0,\ldots,\ell-1\}$, and then $t_{\ell} = 1$.\footnote{{\`A la} \cref{thm:antisymmetric-lp}, this creates $L = 2\ell+1 = 503$ bias classes, labeled $I_{-\ell},\ldots,I_{+\ell}$. In particular, $I_0 = [0,0]$, $I_{+i} = (\frac{i-1}{\ell-1}b,\frac{i}{\ell-1}b]$ for $i \in [\ell-1]$, and $i_{+\ell} = (b,1]$, and $I_{-i}$ is defined symmetrically.} We set $p_i = \frac{\RoundPLSig_b(\frac{i-1}{\ell-1})+\RoundPLSig_b(\frac{i}{\ell-1})}2$ for $i \in [\ell-1]$, and $p_{\ell} = 1$. Evaluating the linear program in \cref{thm:antisymmetric-lp} using the code in \href{https://github.com/singerng/oblivious-csps/blob/main/figures/new_algorithm.py}{\texttt{figures/new\_algorithm.py}} in the source repository, we deduce $\alpha(\Obl{\RoundPLSig_b}) \geq 0.485275$.
\end{proof}

\cref{fig:discretization} below depicts the effect of increasingly fine discretization on approximation ratio. We suspect that even further improvements are possible by running the LP for an even finer discretization, but we do not know how to calculate the limit of infinitely fine discretization (i.e., the value of the actual step function).

\begin{figure}[t!]
\centering
\centering
\begin{tikzpicture}
\begin{axis}[
    width=0.7\textwidth,
    xmin=15, xmax=185,
    ymin=0.481, ymax=0.486,
    legend pos=south east,
    ymajorgrids=true,
    grid style=dashed,
    yticklabel style={/pgf/number format/.cd,fixed,precision=3},
]
\addplot[
color=blue,
mark=*,
]
coordinates {
(21,0.48146963515320246)
(26,0.48230862190277046)
(31,0.48286615857269166)
(36,0.4832838645305919)
(41,0.4835778991571159)
(46,0.4838125624254632)
(51,0.48400830923060223)
(56,0.484155415561258)
(61,0.48428574913842404)
(66,0.48439634968745304)
(71,0.4844861046992611)
(76,0.4845691122920202)
(81,0.4846392191975384)
(86,0.4847004759629851)
(91,0.4847580717865758)
(96,0.48480589604747737)
(101,0.4848506835384357)
(101,0.4848506835384357)
(111,0.48492755819072697)
(121,0.48499272587885456)
(131,0.48504691697997065)
(141,0.48509294540985626)
(151,0.4851334408237919)
(161,0.4851693891081909)
(171,0.4852001349379196)
(181,0.4852276143315759)
};
\addlegendentry{$\RoundPLSig_{149/309}$};
\addplot[
color=red,
mark=square*,
]
coordinates {
(21,0.4811046511627907)
(26,0.48187106918239)
(31,0.4824218750000001)
(36,0.48283884660421544)
(41,0.48316176470588235)
(46,0.48339120370370375)
(51,0.4835683741844472)
(56,0.48372183372183375)
(61,0.48385634527793275)
(66,0.48397435897435903)
(71,0.4840604026845638)
(76,0.48413803721174)
(81,0.4842095588235294)
(86,0.4842757391251047)
(91,0.48433421047847586)
(96,0.4843798853569568)
(101,0.48442358158384274)
(111,0.4845046593367057)
(121,0.4845669934640523)
(131,0.48462235228539574)
(141,0.4846698750851436)
(151,0.4847095908898257)
(161,0.4847472458470999)
(171,0.484776763891152)
(181,0.48480588575886313)
};
\addlegendentry{$\RoundPLSig_{1/2}$};
\end{axis}
\end{tikzpicture}
\caption{A plot depicting how the fineness of discretization affects the approximation ratio calculated by the linear program of \cref{thm:antisymmetric-lp}, for two continuous selection functions: $\RoundPLSig_{1/2}$ and $\RoundPLSig_{149/309}$. Each point represents the approximation ratio of some oblivious algorithm, as calculated by the linear program in \cref{thm:antisymmetric-lp}. The horizontal axis records the number of bias classes (up to sign, i.e., as in \cref{thm:antisymmetric-lp}), and the vertical axis records the calculated approximation ratio. This plot was produced by \href{https://github.com/singerng/oblivious-csps/blob/main/figures/discretization.py}{\texttt{figures/discretization.py}} in the source code.}\label{fig:discretization}
\end{figure}
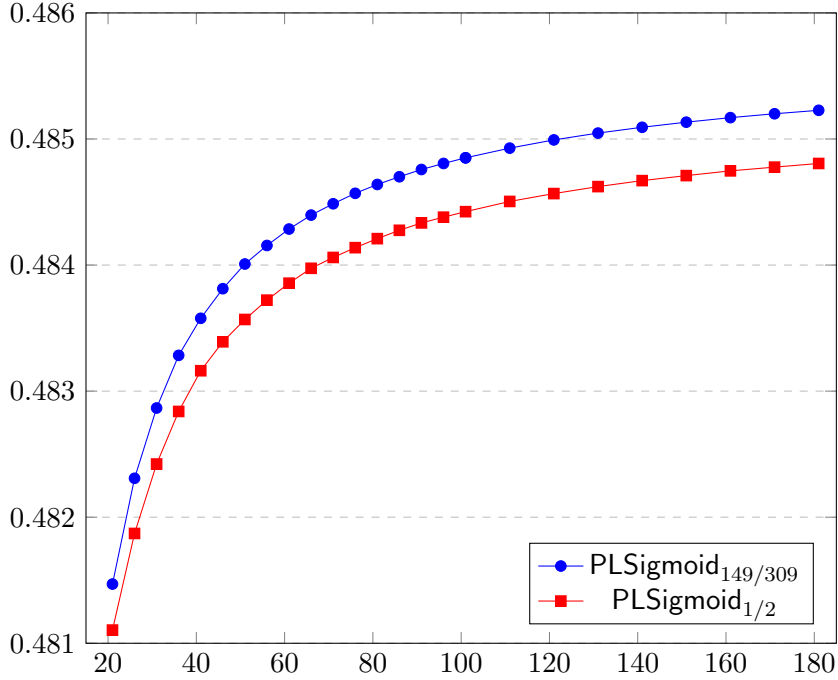

\section{Lower bounds}

In this section, we present a number of lower bounds against oblivious algorithms. These bounds hold at varying levels of ``granularity'': Some are only against individual selection functions (like $\RoundPLSig_{1/2}$, constructed in the last section), while others hold against all PL sigmoid functions, all antisymmetric functions, or even all functions. However, the proofs of these theorems are united by an underlying methodology for producing candidate lower bound graphs, which combines grid search and linear programming. We discuss this methodology next before turning to the proofs of specific bounds.

\subsection{Methodology for finding hard instances}\label{sec:lb-method}

First, we describe our general strategy in this section for finding graphs that are hard for oblivious algorithms. This description is not strictly necessary for our proofs: One may simply take the explicit graphs we produce and verify that they have the desired lower bound properties. Also, some of the graphs are simple enough that one could have invented them by hand. However, we include this description to shed some light on where the more complex graphs come from. Also, when we can indeed prove such ``uniform'' lower bounds, the resulting quantitative bound will hold even against ensembles of oblivious algorithms.

Ideally, given a set $\CC$ of oblivious algorithms, we would like to produce a graph $G$ and show that all algorithms in $\CC$ perform poorly on $G$. This is, in some sense, ``dual'' to the problem of finding good oblivious algorithms, where we want a single oblivious algorithm that performs well on all graphs. In general, we do not have strong evidence as to whether or not this dual procedure is \emph{tight}, i.e., whether we can certify optimal bounds on the performance of algorithms for the classes we're interested in by coming up with a single hard graph $G$. Indeed, some of our proofs instead consider a small set of graphs and prove that any oblivious algorithm in $\CC$ must perform poorly on at least one of these algorithms; we do not know whether this potentially stronger proof method is tight, either.

Let $L \in \BN$, $-1 \leq t_1 < \cdots < t_L \leq +1$, and $\CP \subseteq [0,1]^L$. We will be interested in graphs that contain only vertices of biases $t_1,\ldots,t_L$. For any such graph $G$, we can view a vector $\vecp \in \CP \subseteq [0,1]^L$ as an oblivious algorithm: Each vertex of bias $i$ is assigned to $1$ w.p. $p(i)$. Our LP will output the graph which is ``worst-case'' for all algorithms in $\CP$ (i.e., such that the performance of the best algorithm in $\CP$ is minimized), among all graphs with vertices of bias $t_1,\ldots,t_L$.

\begin{tcolorbox}\begin{empheq}[left=\empheqlbrace]{alignat*=3}
    & \underset{\{w(v_1,v_2) : v_1,v_2 \in V\}}{\mathrm{minimize}} \quad && \max_{\vecp \in \CP} \sum_{v_1 = (b_1,i_1), v_2 = (b_2,i_2) \in V} p(i_1) (1-p(i_2)) w(v_1,v_2) && \\
    & \mathrm{s.t.} && w(v_1,v_2) \geq 0 && \forall v_1,v_2 \in V \\
    & && \sum_{i_1,i_2 \in [L]} w((1,i_1),(0,i_2)) = 1 && \\
    & && t_i (W^+(i) + W^-(i)) = W^+(i) - W^-(i) && \forall v = (b,i) \in V
\end{empheq}
\end{tcolorbox}

where we define the linear functions
\[ W^+(v_1) = \sum_{v_2 \in V} w(v_1,v_2) \text{ and } W^-(v_1) = \sum_{v_2 \in V} w(v_2,v_1). \] (While the objective, as written, is the maximum of a finite number of linear functions, this can be converted to a standard-form LP by using the standard trick which introduces one additional variable.) Note that every feasible solution to this linear program is a weighted, directed graph on $V$ where:
\begin{enumerate}
    \item The cut from $\{1\} \times [\pm\ell]$ to $\{0\} \times [\pm\ell]$ satisfies weight $1$.
    \item The vertices $(0,i)$ and $(1,i)$ have bias $t_i$ (if they are nonisolated) for $i \in [\pm\ell]$.
    \item For $\vecp \in [0,1]^L$, the quantity $\sum_{v_1 = (b_1,i_1), v_2 = (b_2,i_2) \in V} p(i_1) (1-p(i_2)) w(v_1,v_2)$ calculates the expected weight cut by the oblivious algorithm $\vecp$.
\end{enumerate}
That this suffices to produce the worst-case graph is essentially a consequence of some reasoning due \textcite{FJ15} for a similar LP; the idea is to ``condense'' vertices into $2L$ equivalence classes based purely on their bias and their assignment in an optimal cut, and also to rescale so that the optimal cut has weight $1$.

However, to find lower bounds against concrete classes of selection functions, two issues remain. Firstly, we need to choose the biases $b_1,\ldots,b_L$. We typically do this by either fixing a small value of $L$ and performing a grid search, or by picking uniformly spaced points. Secondly, in the case where we desire to prove a bound against a large or infinite set $\CP^*$ of selection functions --- e.g., the set of all antisymmetric selection functions, or PL sigmoid selection functions --- we need to reduce the set of selection functions to a tractable size. We typically do this by discretizing the space of all selection functions in some form to form some small subset $\CP$. This does cause some additional loss in the approximation ratio, which will ideally be small. To carry out the proof, we first produce a candidate graph $G$ which holds against $\CP$, and then (ideally) show that it holds almost as well against $\CP^*$ by directly solving the maximization problem over $\CP^*$. In the case where $\CP^*$ is the set of all antisymmetric functions (or all functions), maximizing the weight of the assignment over $\CP^*$ corresponds to solving an $L$-variate quadratic optimization problem.

\subsection{Bounds against $\RoundPLSig_{1/2}$ (\cref{thm:intercept-1/2-lb})}\label{sec:plsigmoid_half}

We have the following bound on $\RoundPLSig_{1/2}$:

\thmhalf*

\def\FIGpackx{3}
\def\FIGpacky{3}
\def\FIGbend{25}
\def\FIGsubfigwidth{0.45\textwidth}

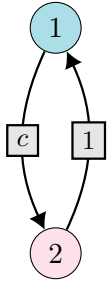
\begin{figure}[h]
\centering
\begin{tikzpicture}
    \node[draw, circle, fill=\FIGcolpos] (v1) at (0,\FIGpacky) {1};
    \node[draw, circle, fill=\FIGcolneg] (v2) at (0,0) {2};

    \path[diredge] (v1) edge[bend right=\FIGbend] node[weight] {$c$} (v2);
    \path[diredge] (v2) edge[bend right=\FIGbend] node[weight] {$1$} (v1);
\end{tikzpicture}
\caption{\emph{Parameter:} $c > 1$. The \FIGcolposdesc vertex ($1$) has bias $+\frac{c-1}{c+1}$. The \FIGcolnegdesc vertex ($2$) has bias $-\frac{c-1}{c+1}$. The assignment $1\to1,2\to0$ satisfies weight $c$. An oblivious assignment $1\to p,2 \to q$ satisfies weight $p (1-q) c + q (1-p)$. (Every two-vertex graph (without self-loops and with two nontrivial edges) is isomorphic to this graph up to rescaling.)}\label{fig:twovertex}
\end{figure}

\begin{proof}
Let $G$ denote the graph in \cref{fig:twovertex} with $1 < c < 3$ TBD. Vertex $1$ has bias $+\frac{c-1}{c+1}$. Since $c < 3$, $\frac{c-1}{c+1} < \frac12$. Thus, $\RoundPLSig_{1/2}$ assigns vertex $1$ to $1$ w.p. $p := \frac12+\frac{c-1}{c+1}$, and $2$ to $1$ w.p. $1-p$, satisfying weight \[ p^2 c + q^2 = \left(\frac12+\frac{c-1}{c+1}\right)^2 c + \left(\frac12-\frac{c-1}{c+1}\right)^2. \] Thus, $\RoundPLSig_{1/2}$ achieves ratio at most \[ p^2 + q^2 c^{-1} = \left(\frac12+\frac{c-1}{c+1}\right)^2 + \left(\frac12-\frac{c-1}{c+1}\right)^2 c^{-1} \] for all $1 < c < 3$. Setting\footnote{This value of $c$ gives the minimum possible bound (though this fact is not necessary).} $c = (9+12\sqrt{2})/23$, the expression is equal to $6\sqrt 2- 8 \approx 0.485282$, as desired. 
\end{proof}

\subsection{Lower bound for PL sigmoid functions (\cref{thm:pl-lb})}\label{sec:lowerbound_plsigmoid}

\thmplsiglb*

\begin{proof}[Proof of \cref{thm:pl-lb}]
We prove this theorem by considering three separate cases based on the value of the intercept $b$ of $\RoundPLSig_b$:

\paragraph*{Case 1: $1/2\le b\le 1$.} For this case, we consider the graph from \cref{thm:intercept-1/2-lb} with $c = 1.12916$.

It follows from the proof of \cref{thm:intercept-1/2-lb} that $\RoundPLSig_{1/2}$ achieves an approximation ratio of at most $0.485282$ on this graph. The approximation ratio of $\RoundPLSig_b$ on this graph is given by
    $$
    \frac{1.12916p^2+(1-p)^2}{1.12916} \, ,
    $$
where $p=\RoundS(b_1)$ and $b_1 = 0.12916/2.12916$.
We now argue that $1/2\le b \le 1$ implies that $\RoundPLSig_b$ gets an approximation ratio of at most $0.485282$ on this graph. Observe that the derivative of $1.12916p^2+(1-p)^2$ is $0$ at $1/2.12916\approx 0.46967$. Hence in the range $0.46967\le p \le 1$, $1.12916p^2+(1-p)^2$ is an increasing function of $p$. By the definition of $\RoundPLSig_b$, we have that $p = 1/2 + b_1/2b$ when $b\ge 1/2$. Hence, for $1/2\le b \le 1$, the largest possible approximation is achieved at $b=1/2$ and we already established that this is at most $0.485282$.
 
\paragraph*{Case 2: $0.225\le b\le 1/2$.} Let $b_i = -0.475 + 0.05(i-1)$ for all $i \in \{1, \dots, 20\}$. \newcommand{\GLp}{G_{\mathrm{LP}}}

We describe a concrete graph $\GLp$, which we visualize in the next figure.\footnote{This graph was originally calculated by the LP paradigm described in \cref{sec:lb-method}. $\CP$ consisted of the single rounding function $\RoundPLSig_{149/309}$. The particularly simple structure of the found solution let us describe it analytically, as we do here.} $\GLp$ has $36$ vertices, which we label $\{1,\ldots,18\}$ and $\{3',\ldots,20'\}$. The graph has the property that for each $i \in \{1,\ldots,20\}$, the vertices $i$ and $i'$ have bias $b_i$. (We treat the biases of the nonexistent vertices $1'$, $2'$, $19$, and $20$ as vacuous.) Further, $\GLp$ has the following very simple unweighted edge structure:
\begin{enumerate}
    \item For each $i \in \{1,\ldots,18\}$, there is an edge $i \to (i+2)'$.
    \item For each $i \in \{1,\ldots,17\}$, there is an edge $(i+3)' \to i$.
    \item There are edges $3' \to 1$ and $20'$ to $18$. (Note that $3'$ and $20'$ have no edges of type (2).)
\end{enumerate}

One may verify that the bias constraints and the edge structure together determine every edge weight up to vertex rescaling: E.g., if a vertex $v$ has bias $b$ and indegree $w$, then it must have outdegree $w\frac{1+b}{1-b}$.

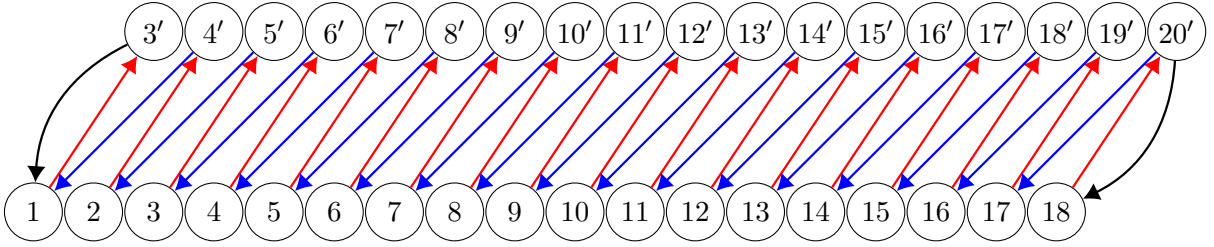
\begin{figure}[t!]
\centering
\begin{tikzpicture}[>=stealth, scale=0.8]
    \foreach \i in {1,...,18} {
        \node[draw, circle, inner sep=2pt, minimum size=2em] (v0\i) at (\i, 0) {\(\i\)};
    }

    \foreach \i in {3,...,20} {
        \node[draw, circle, inner sep=2pt, minimum size=2em] (v1\i) at (\i, 3) {\(\i'\)};
    }

    Draw directed edges (0, i) -> (1, i+2)
    \foreach \i [evaluate=\i as \j using int(\i+2)] in {1,...,18} {
        \draw[diredge, red] (v0\i) -- (v1\j);
    }

    \foreach \i [evaluate=\i as \j using int(\i-3)] in {4,...,20} {
        \draw[diredge, blue] (v1\i) -- (v0\j);
    }

    \path[diredge, black] (v13) edge[bend right=30] (v01);
    \path[diredge, black] (v120) edge[bend left=30] (v018);
\end{tikzpicture}
\caption{The (unweighted version) of the graph used in Case 2 of the proof of \Cref{thm:pl-lb}.}
\end{figure}

Since $\GLp$ is fixed, for any intercept $b \in (0,1]$, $\val_{\GLp}(\RoundPLSig_b)$ is a quadratic function of the selection probabilities $\{\RoundPLSig_b(v)_{v \in V(G)}\}$. Further, recall that vertices in $\GLp$ have biases $\{b_1,\ldots,b_{20}\}$, so we are interested in the selection probabilities $\{\RoundPLSig_b(b_i)\}_{i \in [20]}$. For each bias $b_i$, if $b_i \leq -b,$ then $\RoundPLSig_b(b_i) = 0$ and if $b_i \geq b,$ then $\RoundPLSig_b(b_i) = 1.$ Otherwise, we can write
\[
\RoundPLSig_b(b_i) = \frac{b_i+b}{2b} = \frac12 + \frac{b_i}{2b},
\]
a linear function in $b^{-1}$. Hence over any interval $b \in [b_i,b_{i+1}]$, $\val_{\GLp}(\RoundPLSig_b)$ is a quadratic function of $b^{-1}$. We can therefore find the maximum $b^*$ of this quadratic function; if $b^*$ is in the interval, we are done, otherwise we take the maximum value over $b \in \{b_i,b_{i+1}\}$.

We perform the explicit calculations for the edge-weights and the oblivious ratio in Mathematica (\href{https://github.com/singerng/oblivious-csps/blob/main/plsigmoid_lb.nb}{\texttt{plsigmoid\_lb.nb}} in the source repository) since the rational numbers have many digits. The maxima over the intervals $[b_{20},1/2]$, $[b_{19}, b_{20}]$, $[b_{18},b_{19}]$, $[b_{17},b_{18}]$, $[b_{16},b_{17}]$, and $[b_{15},b_{16}]$ rounded up in the sixth decimal place are, respectively, $0.485895$, $0.485870$, $0.485488$, $0.484375$, $0.482019$, and $0.477739$; all are less than $0.486$.

\paragraph*{Case 3: $0\le b\le 0.225$.} Now, consider the graph shown in Figure \ref{fig:fourvertex} with $c=\frac{1+b}{1-b}$. Vertices $1$ and $2$ have bias $b$ and vertices $3$ and $4$ have bias $-b$. By the definition of $\RoundPLSig_b$, it assigns $\{1,2\}\to 1, \{3,4\}\to 0$ and the corresponding cut has value $c^2-1$. On the other hand, the assignment $\{1,3\} \to 1, \{2,4\}\to0$ satisfies weight $c^2-1 + 2 \cdot 1= c^2+1$. Hence, the approximation ratio is at most $(c^2-1)/(c^2+1).$ For $b\le 0.225$, this value is less than $0.486$.
    
We conclude that no PL sigmoid function can achieve an approximation ratio of $0.486,$ as desired.
    \end{proof}

\section{Lower bound for arbitrary selection functions (\cref{thm:gen-lb})}\label{sec:lowerbound_general}

In this section, we prove \cref{thm:gen-lb}:

\thmgenlb*

This gives a lower bound on the ratio achievable by \emph{any} (not necessarily antisymmetric) selection function. We use the graph in \cref{fig:twovertex} above as well as \cref{fig:fourvertex} below:

\def\FIGpackx{3}
\def\FIGpacky{3}
\def\FIGbend{25}
\def\FIGsubfigwidth{0.45\textwidth}

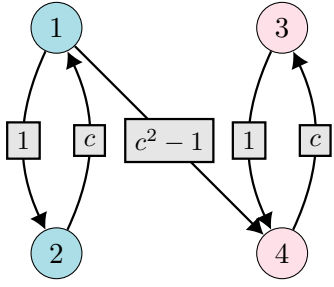
\begin{figure}[h]
\centering
\begin{tikzpicture}
    \node[draw, circle, fill=\FIGcolpos] (v3) at (\FIGpackx,\FIGpacky) {1};
    \node[draw, circle, fill=\FIGcolpos] (v4) at (\FIGpackx,0) {2};
    \node[draw, circle, fill=\FIGcolneg] (v5) at (2*\FIGpackx,\FIGpacky) {3};
    \node[draw, circle, fill=\FIGcolneg] (v6) at (2*\FIGpackx,0) {4};
    
    \path[diredge] (v3) edge[bend right=\FIGbend] node[weight] {$1$} (v4);
    \path[diredge] (v4) edge[bend right=\FIGbend] node[weight] {$c$} (v3);
    \path[diredge] (v5) edge[bend right=\FIGbend] node[weight] {$1$} (v6);
    \path[diredge] (v6) edge[bend right=\FIGbend] node[weight] {$c$} (v5);
    \path[diredge] (v3) edge node[weight] {$c^2-1$} (v6);
\end{tikzpicture}
\caption{\emph{Parameter:} $c>1$. The \FIGcolposdesc vertices ($\{1,2\}$) have bias $+\frac{c-1}{c+1}$ (note that $\frac{(c^2-1)+1-c}{(c^2-1)+1+c} = \frac{c(c-1)}{c(c+1)} = \frac{c-1}{c+1}$). The \FIGcolnegdesc vertices ($\{3,4\}$) have bias $-\frac{c-1}{c+1}$. The assignment $\{1,3\} \to 1, \{2,4\}\to0$ satisfies weight $c^2-1 + 2 \cdot 1= c^2+1$. An oblivious assignment $\{1,2\} \to p,\{3,4\}\to q$ satisfies weight $p(1-p)(c+1) + q(1-q) (c+1) + p(1-q) (c^2-1)$.}\label{fig:fourvertex}
\end{figure}

\begin{proof}
    Let $0 \leq \lambda < 1$ and $c > 1$ be parameters. Consider a graph $G$ consisting of disjoint copies of the graphs in \cref{fig:twovertex,fig:fourvertex} weighted by $\lambda$ and $1-\lambda$, respectively. Examining these graphs, we see that all vertices in $G$ have bias $\pm \frac{c-1}{c+1}$. Further, there exists a cut cutting weight
    \begin{equation}\label{eq:gen-lb:opt-wt}
        \lambda c + (1-\lambda) (c^2+1)
    \end{equation} while an oblivious assignment assigning the positive-bias vertices to $1$ w.p. $p$ and the negative-bias vertices to $1$ w.p. $q$ achieves value
    \begin{equation}\label{eq:gen-lb:obl-wt}
    \lambda (p (1-q) c + q (1-p)) + (1-\lambda) p(1-p)(c+1) + q(1-q) (c+1) + p(1-q) (c^2-1).
    \end{equation}

    At $\lambda=\frac{15}{32}$ and $c = \frac98$,\footnote{These fractions were suggested by numerical search on a computer; they are almost certainly non-optimal, but we include fractions so that the calculations in this proof can be checked exactly. For simplicity's sake, we mention that a ratio strictly below $\frac12$ (and indeed, better than that of \cref{thm:fj-lb-general}) is achieved by redoing this calculation at $\lambda=\frac13$ and $c=\frac54$.} the quantity in \cref{eq:gen-lb:obl-wt} is optimized at $(p,q)=(\frac{1352}{2295},\frac{943}{2295})$. Dividing this value by the value of \cref{eq:gen-lb:opt-wt} at the same values of $\lambda$ and $c$, we arrive at a ratio of $\frac{4031104}{8135775}\approx 0.4955$.
\end{proof}

\section{Lower bounds for antisymmetric selection functions (\cref{thm:antisym-lb})}\label{sec:lowerbound_antisymmetric}

In this section, we prove \cref{thm:antisym-lb}:

\thmsymlb*

This theorem improves on the bound of \textcite{FJ15} (see \cref{thm:fj-lb-antisym} above). The proof uses the following single graph:

\def\xsep{4.5}
\def\ysep{5}

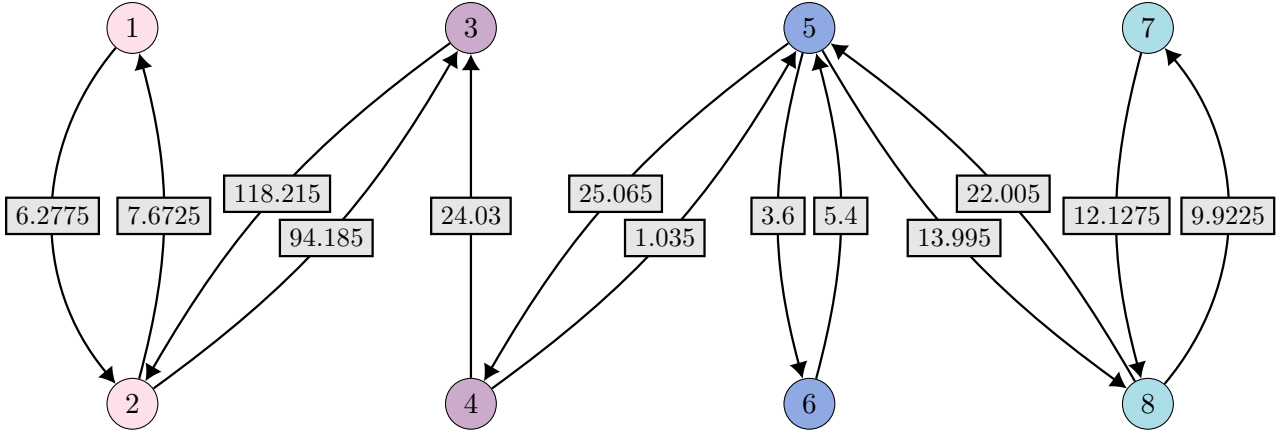
\begin{figure}
\centering
\begin{tikzpicture}
    \node[draw, circle, fill=\FIGcolneg] (1) at (0, \ysep) {$1$};
    \node[draw, circle, fill=\FIGcolzero] (3) at (\xsep, \ysep) {$3$};
    \node[draw, circle, fill=\FIGcolpos] (7) at (3*\xsep, \ysep) {$7$};
    \node[draw, circle, fill=\FIGcolsuperpos] (5) at (2*\xsep, \ysep) {$5$};
    \node[draw, circle, fill=\FIGcolneg] (2) at (0, 0) {$2$};
    \node[draw, circle, fill=\FIGcolzero] (4) at (\xsep, 0) {$4$};
    \node[draw, circle, fill=\FIGcolpos] (8) at (3*\xsep, 0) {$8$};
    \node[draw, circle, fill=\FIGcolsuperpos] (6) at (2*\xsep, 0) {$6$};

    \draw[diredge] (1) to[bend right=40] node[weight] {6.2775} (2);
    \draw[diredge] (2) to[bend right=15] node[weight] {7.6725} (1);
    \draw[diredge] (7) to[bend right=15] node[weight] {12.1275} (8);
    \draw[diredge] (8) to[bend right=40] node[weight] {9.9225} (7);
    \draw[diredge] (5) to[bend right=15] node[weight] {3.6} (6);
    \draw[diredge] (6) to[bend right=15] node[weight] {5.4} (5);
    \draw[diredge] (2) to[bend right=12] node[weight] {94.185} (3);
    \draw[diredge] (3) to[bend right=12] node[weight] {118.215} (2);
    \draw[diredge] (8) to[bend right=12] node[weight] {22.005} (5);
    \draw[diredge] (5) to[bend right=12] node[weight] {13.995} (8);
    \draw[diredge] (4) to[bend right=12] node[weight] {1.035} (5);
    \draw[diredge] (5) to[bend right=12] node[weight] {25.065} (4);
    \draw[diredge] (4) -- node[weight] {24.03} (3);
\end{tikzpicture}
\caption{Graph giving an improved lower bound against antisymmetric selection functions. The \FIGcolnegdesc vertices ($1$ and $2$), the \FIGcolzerodesc vertices ($3$ and $4$), the \FIGcolposdesc vertices ($7$ and $8$), and the \FIGcolsuperposdesc vertices ($5$ and $6$) have biases $-0.1$, $0$, $+0.1,$ and $+0.2$, respectively. The cut assigning $\{1,3,5,7\} \to 1$ and $\{2,4,6,8\} \to 0$ has weight $6.2775+118.215+25.065+3.6+13.995+12.1275 = 179.28$. An antisymmetric oblivious cut assigning $\{1,2\}\to 1-p$, $\{3,4\} \to \frac12$, $\{5,6\} \to q$, and $\{7,8\} \to p$ has value $(6.2775 + 7.6725)p(1-p) + 94.185 \cdot \frac12(1-p)+ 118.215 \cdot \frac12 p + 24.03 \cdot \frac14 + 1.035 \cdot \frac12 (1-q) + 25.065\cdot \frac12 q + 22.005 p(1-q)+13.995q(1-p)+ (3.6+5.4) q(1-q) + (12.1275+9.9225) \cdot p(1-p) = 53.6175 - 36 p^2 + p (70.02 - 36 q) + 35.01 q - 9q^2$.
}
\label{4889figure}
\end{figure}

\begin{proof}
    Let $G$ denote the graph in \cref{4889figure}.\footnote{We found the graph in \cref{4889figure} by employing the linear programming methodology described in \cref{sec:lb-method}. We considered $L=5$ possible bias classes, $\{\pm b_1, \pm b_2, 0\}$, where $b_1$ and $b_2$ were multiples of $\frac1{10}$, and performed a grid search over these possibilities. To discretize the space of all antisymmetric functions, we considered all functions mapping $b_1$ and $b_2$ to multiples of $\frac1{100}$ between $\frac12$ and $1$.} As described in the figure's caption, this graph has a cut of weight $179.28$, while any oblivious cut has value $v(p,q) = 53.6175 - 36 p^2 + p (70.02 - 36 q) + 35.01 q - 9q^2$ where $0\leq p,q\leq 1$ are assignment probabilities for the two bias classes. Since $\frac{\partial v}{\partial p} = 70.02 - 72 p - 36 q$ and $\frac{\partial v}{\partial q} = 35.01 - 36 p - 18 q$, both derivatives vanish when $q = 1.945-2p$. Substituting back, we see that $v(p,1.945-2p) = 87.6647$; this is the minimum value of $v(p,q)$ over $\BR^2$ (and it is achieved at e.g. $(p,q)=(0.6,0.745) \in [0,1]^2$). Finally, $\frac{87.6647}{179.28} \approx 0.48898$, as desired.
\end{proof}

\newpage

\appendix
\section{Recap: The prior lower bound of \textcite{FJ15} (\cref{thm:fj-lb-antisym})}\label{sec:fj-lb}

In this appendix, we include a brief description and analysis of the graph used to prove the lower bound of \textcite{FJ15} (\cref{thm:fj-lb-antisym}):

\thmfjantisym*

To prove this theorem, \textcite{FJ15} used the pair of graphs appearing in \cref{fig:fj-graphs} below:

\begin{figure}[h]
\centering
\begin{subfigure}[t]{\FIGsubfigwidth}
\centering
\begin{tikzpicture}
    \node[draw, circle, fill=\FIGcolpos] (v1) at (0,\FIGpacky) {1};
    \node[draw, circle, fill=\FIGcolpos] (v2) at (0,0) {2};
    \node[draw, circle, fill=\FIGcolzero] (v3) at (\FIGpackx,\FIGpacky) {3};
    \node[draw, circle, fill=\FIGcolzero] (v4) at (\FIGpackx,0) {4};
    \node[draw, circle, fill=\FIGcolneg] (v5) at (2*\FIGpackx,\FIGpacky) {5};
    \node[draw, circle, fill=\FIGcolneg] (v6) at (2*\FIGpackx,0) {6};

    \path[diredge] (v1) edge[bend right=\FIGbend] node[weight] {1} (v2);
    \path[diredge] (v2) edge[bend right=\FIGbend] node[weight] {$c$} (v1);
    \path[diredge] (v1) edge node[weight] {$c^2-1$} (v4);
    \path[diredge] (v4) edge node[weight] {$c^2-1$} (v3);
    \path[diredge] (v3) edge node[weight] {$c^2-1$} (v6);
    \path[diredge] (v5) edge[bend right=\FIGbend] node[weight] {1} (v6);
    \path[diredge] (v6) edge[bend right=\FIGbend] node[weight] {$c$} (v5);
\end{tikzpicture}
\caption{\emph{Parameter:} $c>1$ (``$G_1$'' from \cite{FJ15}). The \FIGcolposdesc vertices ($\{1,2\}$) have bias $+(c-1)/(c+1)$. The \FIGcolzerodesc vertices ($\{3,4\}$) have bias $0$. The \FIGcolnegdesc vertices ($\{5,6\}$) have bias $-(c-1)/(c+1)$. The assignment $\{1,3,5\} \to 1, \{2,4,6\}\to0$ satisfies weight $2\cdot (c^2-1) + 2 \cdot 1= 2c^2$. An oblivious assignment $\{1,2\} \to p,\{3,4\}\to q, \{5,6\} \to r$ satisfies weight $p(1-p)(c+1)+p(1-r)(c^2-1)+r(1-r)(c^2-1)+r(1-q)(c^2-1)+q(1-q)(c+1)$.}\label{fig:fj-graphs:a}
\end{subfigure}
~
\begin{subfigure}[t]{\FIGsubfigwidth}
\centering
\begin{tikzpicture}
    \node[draw, circle, fill=\FIGcolzero] (v7) at (3*\FIGpackx,\FIGpacky) {1};
    \node[draw, circle, fill=\FIGcolpos] (v8) at (3*\FIGpackx,0) {2};
    \node[draw, circle, fill=\FIGcolneg] (v9) at (4*\FIGpackx,\FIGpacky) {3};
    \node[draw, circle, fill=\FIGcolzero] (v10) at (4*\FIGpackx,0) {4};

    \path[diredge] (v7) edge[bend right=\FIGbend] node[weight] {1} (v8);
    \path[diredge] (v8) edge[bend right=\FIGbend] node[weight] {$c$} (v7);
    \path[diredge] (v9) edge[bend right=\FIGbend] node[weight] {1} (v10);
    \path[diredge] (v10) edge[bend right=\FIGbend] node[weight] {$c$} (v9);
    \path[diredge] (v7) edge node[weight] {$c-1$} (v10);
\end{tikzpicture}
\caption{\emph{Parameter:} $c>1$ (``$G_2$'' from \cite{FJ15}). The \FIGcolposdesc vertex ($2$) has bias $+(c-1)/(c+1)$. The \FIGcolzerodesc vertices ($\{1,4\}$) have bias $0$. The \FIGcolnegdesc vertex ($3$) has bias $-(c-1)/(c+1)$. The assignment $\{1,3\} \to 0, \{2,4\}\to1$ satisfies weight $2c$. An oblivious assignment $2 \to p, 3 \to q, \{1,4\} \to r$ satisfies weight $p(1-r)(c) + r(1-p) + r(1-r)(c-1) + r(1-q)(c) + q(1-r)$.}\label{fig:fj-graphs:b}
\end{subfigure}
\caption{The pair of graphs used to prove a lower bound against antisymmetric selection functions by \textcite{FJ15}.}\label{fig:fj-graphs}
\end{figure}

\begin{proof}
    Let $0 \leq \lambda \leq 1$ and $c > 1$ be two TBD constants. Let $G$ denote a weighted disjoint union of the graphs in \cref{fig:fj-graphs:a,fig:fj-graphs:b}, weighted by $\lambda$ and $1-\lambda$, respectively. As in the figures' caption, $G$ has a cut satisfying weight \[ \lambda (2c^2) + (1-\lambda)(2c). \] Now, consider an antisymmetric selection function $\RoundS : [-1,+1] \to [0,1]$. Suppose $\RoundS(+\frac{c-1}{c+1}) = p$. By antisymmetry, $\RoundS(0) = \frac12$ and $\RoundS(-\frac{c-1}{c+1}) = 1-p$. Thus, as in the caption, $\Obl{\RoundS}$ satisfies weight
    \begin{multline*}
        \lambda \left(p(1-p)(c+1)+\frac12p(c^2-1)+\frac14(c^2-1)+\frac12p(c^2-1)+p(1-p)(c+1) \right) \\ + (1-\lambda)\left(\frac12p(c) + \frac12(1-p) + \frac14(c-1) + \frac12 p(c) + \frac12(1-p)\right) \\
        = (1 - \lambda) \left(1 + \left(\frac14 + p\right) (c-1)\right) + \lambda \left(\left(\frac14 + p\right) (c^2-1) + 2 (c+1) (1 - p) p)\right)
    \end{multline*}
    Thus, $\Obl{\RoundS}$ has ratio at most
    \[
    \frac{(1 - \lambda) \left(1 + \left(\frac14 + p\right) (c-1)\right) + \lambda \left(\left(\frac14 + p\right) (c^2-1) + 2 (c+1) (1 - p) p\right)}{2 (\lambda c^2 + (1-\lambda) c)}.
    \]
    Evaluating at $c=\frac54$ and $\lambda = \frac34$ (found in \cite{FJ15} using computer search) and then maximizing over $p$ gives the bound.
\end{proof}

We note that the graphs we used to prove \cref{thm:gen-lb}, i.e., \cref{fig:twovertex,fig:fourvertex}, are simpler than the graphs used in \cite{FJ15} to prove \cref{thm:fj-lb-antisym}, i.e., \cref{fig:fj-graphs}: In particular, there are no vertices of bias zero.

\printbibliography

@inproceedings{CGV20,
  booktitle = {2020 {{IEEE}} 61st {{Annual Symposium}} on {{Foundations}} of {{Computer Science}}},
  author = {Chou, Chi-Ning and Golovnev, Alexander and Velusamy, Santhoshini},
  date = {2020-11},
  pages = {330--341},
  publisher = {{IEEE Computer Society}},
  doi = {10.1109/FOCS46700.2020.00039},
  eventdate = {2020-11-16/2020-11-19},
  eventtitle = {{{FOCS}} 2020},
  venue = {virtual},
  title = {Optimal {{Streaming Approximations}} for All {{Boolean Max-2CSPs}} and {{Max-}}{\(k\)}{{SAT}}}
}

@article{FJ15,
  title = {Oblivious {{Algorithms}} for the {{Maximum Directed Cut Problem}}},
  author = {Feige, Uriel and Jozeph, Shlomo},
  date = {2015-02},
  journaltitle = {Algorithmica},
  shortjournal = {Algorithmica},
  volume = {71},
  number = {2},
  pages = {409--428},
  doi = {10.1007/s00453-013-9806-z}
}

@inproceedings{KK19,
  title = {An Optimal Space Lower Bound for Approximating {{MAX-CUT}}},
  booktitle = {Proceedings of the 51st {{Annual ACM SIGACT Symposium}} on {{Theory}} of {{Computing}}},
  author = {Kapralov, Michael and Krachun, Dmitry},
  date = {2019-06},
  pages = {277--288},
  publisher = {{Association for Computing Machinery}},
  doi = {10.1145/3313276.3316364},
  eventdate = {2019-06-23/2019-06-26},
  eventtitle = {{{STOC}} 2019},
  venue = {Phoenix, AZ, USA}
}

@inproceedings{KKS15,
  title = {Streaming Lower Bounds for Approximating {{MAX-CUT}}},
  booktitle = {Proceedings of the 26th {{Annual ACM-SIAM Symposium}} on {{Discrete Algorithms}}},
  author = {Kapralov, Michael and Khanna, Sanjeev and Sudan, Madhu},
  date = {2015-01},
  pages = {1263--1282},
  publisher = {{Society for Industrial and Applied Mathematics}},
  doi = {10.1137/1.9781611973730.84},
  eventdate = {2015-01-04/2015-01-06},
  eventtitle = {{{SODA}} 2015},
  venue = {San Diego, California, USA}
}

@inproceedings{Sin23-kand,
  booktitle = {Approximation, {{Randomization}}, and {{Combinatorial Optimization}}. {{Algorithms}} and {{Techniques}}},
  author = {Singer, Noah G.},
  editor = {Megow, Nicole and Smith, Adam D.},
  date = {2023-05},
  series = {{{LIPIcs}}},
  volume = {275},
  doi = {10.4230/LIPIcs.APPROX/RANDOM.2023.15},
  eventdate = {2023-09-11/2023-09-13},
  eventtitle = {{{APPROX}} 2023},
  venue = {Atlanta, GA, USA},
  title = {Oblivious Algorithms for the {{Max-}}{\(k\)}{{AND}} Problem}
}

@inproceedings{SSSV23-dicut,
  title = {Improved Streaming Algorithms for {{Maximum Directed Cut}} via Smoothed Snapshots},
  booktitle = {63rd {{Annual Symposium}} on {{Foundations}} of {{Computer Science}}},
  author = {Saxena, Raghuvansh R. and Singer, Noah and Sudan, Madhu and Velusamy, Santhoshini},
  date = {2023},
  publisher = {{IEEE Computing Society}},
  eventdate = {2023-11-06/2023-11-09},
  venue = {Santa Cruz, CA, USA}
}

@inproceedings{SSSV23-random-ordering,
  title = {Streaming Complexity of {{CSPs}} with Randomly Ordered Constraints},
  booktitle = {Proceedings of the 2023 {{Annual ACM-SIAM Symposium}} on {{Discrete Algorithms}}},
  author = {Saxena, Raghuvansh R. and Singer, Noah G. and Sudan, Madhu and Velusamy, Santhoshini},
  date = {2023-01},
  doi = {10.1137/1.9781611977554.ch156},
  eventdate = {2023-01-22/2023-01-25},
  eventtitle = {{{SODA}} 2023},
  venue = {Florence, Italy}
}

@misc{kallaugher2023exponential,
      title={Exponential Quantum Space Advantage for Approximating Maximum Directed Cut in the Streaming Model}, 
      author={John Kallaugher and Ojas Parekh and Nadezhda Voronova},
      year={2023},
      eprint={2311.14123},
      archivePrefix={arXiv},
      primaryClass={quant-ph},
      note = "To appear in STOC 2024"
}

@article{BFS19, author = {Buchbinder, Niv and Feldman, Moran and Schwartz, Roy}, title = {Online Submodular Maximization with Preemption}, year = {2019}, issue_date = {July 2019}, publisher = {Association for Computing Machinery}, address = {New York, NY, USA}, volume = {15}, number = {3}, issn = {1549-6325}, url = {https://doi.org/10.1145/3309764}, doi = {10.1145/3309764}, journal = {ACM Trans. Algorithms}, month = {jun}, articleno = {30}, numpages = {31} }

@thesis{Luk14,
  type = {mathesis},
  title = {Mechanism Design for Boolean Constraint Satisfaction Problems},
  author = {Lukovics, \'Akos},
  date = {2014},
  institution = {{\'Ecole Polytechnique F\'ed\'erale de Lausanne}},
  location = {{ Lausanne, Switzerland}}
}

\end{document}